\documentclass[a4paper, USenglish, numberwithinsect, cleveref]{lipics-v2019}

\usepackage{xparse, xspace, stmaryrd, paralist}

\newcommand{\Nat}{\mathbb N}
\newcommand{\restr}{{\restriction}}
\renewcommand{\phi}{\varphi}
\renewcommand{\epsilon}{\varepsilon}
\newcommand{\Pht}[1]{\mathit{Pht}_{#1}}
\newcommand{\pht}[3]{[#2]_{#1}^{#3}}
\newcommand{\Comp}{\mathit{Comp}}
\newcommand{\arr}{\mathbin{\to}}
\newcommand{\FV}{\mathit{FV}}
\newcommand{\unbound}{\ensuremath{\mathsf U}\xspace}
\newcommand{\wmsou}{{\upshape WMSO+}\unbound}
\newcommand{\msou}{{\upshape MSO+}\unbound}
\newcommand{\mso}{{\upshape MSO}\xspace}
\newcommand{\child}{\downarrow}%\curlywedgedownarrow}
\newcommand{\Pp}{{\mathcal P}}
\newcommand{\Tt}{{\mathcal T}}
\newcommand{\Vv}{{\mathcal V}}
\newcommand{\Xvar}{\mathsf{X}}
\newcommand{\Yvar}{\mathsf{Y}}
\newcommand{\dom}{\mathrm{dom}}
\newcommand{\rMax}{r_{\max}}
\newcommand{\unode}{u}
\newcommand{\wnode}{w}
\newcommand{\VALempty}{\nu_\emptyset}
\newcommand{\true}{\mathsf{tt}}
\newcommand{\false}{\mathsf{ff}}
\newcommand{\set}[1]{\{#1\}}
\newcommand{\roottree}{\mathit{root}}
\newcommand{\Alphabet}{\mathbb{A}}
\newcommand{\Bphabet}{\mathbb{B}}
\newcommand{\phimso}{\phi_\text{MSO}}
\newcommand{\phifo}{\phi_\text{FO}}

\title{Compositionality of the \msou Logic}

\author{Paweł Parys}{Institute of Informatics, University of Warsaw, Poland}{parys@mimuw.edu.pl}{https://orcid.org/0000-0001-7247-1408}{}

\authorrunning{P. Parys}

\Copyright{Paweł Parys}



\funding{Work supported by the National Science Centre, Poland (grant no.\ 2016/\allowbreak 22/\allowbreak E/\allowbreak ST6/\allowbreak 00041).}

\nolinenumbers %uncomment to disable line numbering

\hideLIPIcs  %uncomment to remove references to LIPIcs series (logo, DOI, ...), e.g. when preparing a pre-final version to be uploaded to arXiv or another public repository

\begin{document}

\maketitle

\begin{abstract}
	We prove that the \msou logic is compositional in the following sense: 
	whether an \msou formula holds in a tree $T$ depends only on \msou-definable properties of the root of $T$ and of subtrees of $T$ starting directly below the root.
	Another kind of compositionality follows: every \msou formula whose all free variables range only over finite sets of nodes 
	(in particular, whose all free variables are first-order) can be rewritten into an \mso formula having access to properties of subtrees definable by \msou sentences (without free variables).
\end{abstract}

\section{Introduction}\label{sec:intro}

	The \msou logic extends the \mso logic by the unbounding quantifier, \unbound~\cite{BojanczykU}.
	A formula using this quantifier, $\unbound \Xvar.\phi$, says that $\phi$ holds for arbitrarily large finite sets $\Xvar$.
	In this paper, we consider \msou formulae evaluated over infinite trees.
	
	The \msou logic was shown to be undecidable, already over infinite words~\cite{mso+u-undecid}.
	Nevertheless, some its fragments have decidable properties.
	Among them there is the weak fragment, \wmsou, where one can only quantify over finite sets~\cite{wmso+u-words,wmso+u-kaiser,wmso+u-trees}.
	The weak fragment can be also extended by the ``exists a branch'' quantifier~\cite{wmso+up}.
	Another fragment, decidable over infinite words, corresponds to $\omega BS$-automata~\cite{omegaBS}.
	
	In a previous paper~\cite{wmsou-schemes}, we show that \wmsou formulae can be evaluated over trees generated by higher-order recursion schemes.
	As an ingredient, we use there compositionality of the \wmsou logic.
	In the current note, we extend the compositionality results to the full \msou logic.

	Namely, we prove two facts.
	Firstly, we show that whether an \msou formula holds in a tree $T$ depends only on \msou-definable properties of the root of $T$ and of subtrees of $T$ starting directly below the root.
	Secondly, every \msou formula whose all free variables range only over finite sets of nodes (in particular, whose all free variables are first-order) 
	can be rewritten into an \mso formula having access to properties of subtrees definable by \msou sentences (without free variables).
	
	Analogous results hold for most logics, and were often used to obtain decidability results 
	(some selection:~\cite{FefermanVaught, Shelah, Lauchli, BlumensathColcombet, wmso+u-kaiser, wmso-model}).

\section{Preliminaries}\label{sec:prelim}

	The powerset of a set $X$ is denoted $\Pp(X)$.
	The domain of a function $f$ is denoted $\dom(f)$.
	When $f$ is a function, by $f[x\mapsto y]$ we mean the function that maps $x$ to $y$ and every other $z\in\dom(f)$ to $f(z)$.

\subparagraph{Trees.}

	We consider ordered trees of bounded arity.
	Fix some maximal arity $\rMax\in\Nat$.
	A \emph{tree domain} (a set of tree nodes) is a set $D\subseteq\set{1,\dots,\rMax}^*$ such that 
	if $\unode i\in D$ then $\unode\in D$, and if $\unode(i+1)\in D$ then $\unode i\in D$ (where $\unode\in\set{1,\dots,\rMax}^*$, $i\in\set{1,\dots,\rMax}$).
	A \emph{tree} over an alphabet $\Alphabet$ is a function $T\colon D\to \Alphabet$, for some tree domain $D$.
	The set of trees over an alphabet $\Alphabet$ and with maximal arity $\rMax$ is denoted $\Tt(A,\rMax)$.
	A node $\wnode$ is the \emph{$i$-th child} of $\unode$ if $\wnode=\unode i$.

\subparagraph{\msou.}

	For technical convenience, we use a syntax in which there are no first-order variables.
	It is easy to translate a formula from a more standard syntax to ours (at least when the maximal arity of considered trees is fixed).
	We assume an infinite set $\Vv$ of variables, which can be used to quantify over sets of tree nodes.
	In the syntax of \msou we have the following constructions:
	\begin{align*}
		\phi::=a(\Xvar)\mid
			 \Xvar\child_i \Yvar\mid
			 \Xvar\subseteq \Yvar\mid
			\phi_1\land\phi_2\mid
			\neg\phi'\mid
			\exists \Xvar.\phi'\mid
			\unbound \Xvar.\phi'
	\end{align*}
	where $a$ is a letter, $i\in\Nat_+$, and $\Xvar,\Yvar\in\Vv$.
	Free variables of a formula are defined as usual; in particular $\unbound\Xvar$ is a quantifier, hence it bounds the variable $\Xvar$.
	By $\FV(\phi)$ we denote the set of free variables of a formula $\phi$.

	The \mso logic is defined likewise, with the exception that the \unbound quantifier is disallowed.

	A \emph{valuation} in a tree $T$ is a function $\nu\colon\Vv\to\Pp(\dom(T))$ 
	(formally, we assume that $\nu$ is defined for all variables from $\Vv$; nevertheless, its value is meaningful only for free variables of a considered formula).
	
	The semantics of a formula $\phi$ in a tree $T$ under a valuation $\nu$ is defined as follows:
	\begin{compactitem}
	\item	$a(\Xvar)$ holds when every node in $\nu(\Xvar)$ is labeled by $a$,
	\item	$\Xvar\child_i \Yvar$ holds when both $\nu(\Xvar)$ and $\nu(\Yvar)$ are singletons,
		and the unique node in $\nu(\Yvar)$ is the $i$-th child of the unique node in $\nu(\Xvar)$,
	\item	$\Xvar\subseteq \Yvar$ holds when $\nu(\Xvar)\subseteq\nu(\Yvar)$,
	\item	$\phi_1\land\phi_2$ holds when both $\phi_1$ and $\phi_2$ hold,
	\item	$\neg\phi'$ holds when $\phi'$ does not hold,
	\item	$\exists\Xvar.\phi'$ holds when $\phi'$ holds under a valuation $\nu[\Xvar\mapsto X]$ for some set $X$ of nodes of $T$, and
	\item	$\unbound\Xvar.\phi'$ holds when for every $n\in\Nat$,
		$\phi'$ holds under a valuation $\nu[\Xvar\mapsto X_n]$ for some finite set $X_n$ of nodes of $T$ of cardinality at least $n$.
	\end{compactitem}
	We write $T,\nu\models\phi$ to denote that $\phi$ holds in $T$ under the valuation $\nu$.
	When $\phi$ is a sentence (i.e., does not have free variables), the valuation $\nu$ is irrelevant, and we simply write $T\models\phi$ instead.

	In order to see that our definition of \msou is not too poor, let us write a few example formulae.
	\begin{compactitem}
	\item	The fact that $\Xvar$ represents an empty set can be expressed as $\mathit{empty}(\Xvar)\equiv\forall\Yvar.\;X\subseteq Y$.
	\item	The fact that $\Xvar$ represents a set of size at least $2$
		can be expressed as $\mathit{big}(\Xvar)\equiv\exists \Yvar.(\Yvar\subseteq\Xvar\land\allowbreak\neg(\Xvar\subseteq \Yvar)\land\neg\mathit{empty}(\Yvar))$.
	\item	The fact that $\Xvar$ represents a singleton can be expressed as $\mathit{sing}(\Xvar)\equiv\neg\mathit{empty}(\Xvar)\land\neg\mathit{big}(\Xvar)$.
	\item	When we only consider trees of a fixed maximal arity $\rMax$,
		the fact that $\Xvar$ and $\Yvar$ represent singletons $\{x\},\{y\}$, respectively, such that $y$ is a child of $x$ can be expressed as
		\begin{align*}
			(\Xvar\child_1 \Yvar)\lor\dots\lor(\Xvar\child_{\rMax} \Yvar)\,,
		\end{align*}
		where $\phi_1\lor\phi_2$ stands for $\neg(\neg\phi_1\land\neg\phi_2)$.
	\item	Let $A=\{a_1,\dots,a_k\}$ be a finite set of letters.
		The fact every node in the set represented by $\Xvar$ has label in $A$ can be expressed as
		\begin{align*}
			\forall\Yvar.\big((\mathit{sing}(\Yvar)\land\Yvar\subseteq\Xvar)\arr(a_1(\Yvar)\lor\dots\lor a_k(\Yvar))\big)\,,
		\end{align*}
		where $\forall\Yvar.\phi$ stands for $\neg\exists\Yvar.\neg\phi$, and $\phi_1\arr\phi_2$ stands for $\neg(\phi_1\land\neg\phi_2)$.
	\end{compactitem}

\subparagraph{Subtrees.}

	For a tree $T$ and its node $\unode$, by $T\restr_\unode$ we denote the \emph{subtree} of $T$ starting at $\unode$, defined in the expected way.
	Moreover, when $\nu$ is a valuation in $T$, by $\nu\restr_\unode$ we denote its restriction to $T\restr_\unode$;
	namely, every variable $\Xvar$ is mapped to the set $\set{\wnode\mid \unode\wnode\in\nu(\Xvar)}$.

	The \emph{root tree} of $T$, denoted $\roottree(T)$, is the tree consisting only of the root of $T$
	(i.e., $\roottree(T)$ consists of a single node labeled by $T(\epsilon)$).
	For a valuation $\nu$, by $\roottree(\nu)$ we denote the appropriate restriction of $\nu$ ;
	it maps every variable $\Xvar$ to the set $\set{\epsilon}\cap\nu(\Xvar)$.

\section{Results}

	Our first theorem says that whether an \msou formula holds in a tree $T$ depends only
	on \msou-definable properties of the root of $T$ and of subtrees of $T$ starting directly below the root.

	\begin{theorem}\label{thm:main-1}
		Fix a finite alphabet $\Alphabet$ and a maximal arity $\rMax$.
		For every \msou formula $\phi$ there exists a finite set $\Omega$ of tuples of \msou formulae
		such that for every tree $T\in\Tt(\Alphabet,\rMax)$ and for every valuation $\nu$ in $T$,
		it is equivalent whether
		\begin{compactitem}
		\item	$T,\nu\models\phi$, and
		\item	for some tuple $(\phi_0,\phi_1,\dots,\phi_r)\in\Omega$ the root of $T$ has $r$ children, 
			and $\roottree(T),\roottree(\nu)\models\phi_0$, and $T\restr_i,\nu\restr_i\models \phi_i$ for all $i\in\set{1,\dots,r}$.
		\end{compactitem}
		Moreover, for every tuple $(\phi_0,\phi_1,\dots,\phi_r)\in\Omega$ and every $i\in\set{0,\dots,r}$
		it is the case that $\FV(\phi_i)\subseteq\FV(\phi)$.
	\end{theorem}
	
	Notice that the formulae $\phi_0$ evaluated in the root are necessarily very simple: 
	they can only read the root's label, and check which variables among $\FV(\phi)$ are mapped to sets containing the root.

	Our second theorem says that every \msou formula whose free variables range only over finite sets of nodes 
	can be rewritten into an \mso formula having access to properties of subtrees definable by \msou sentences.

	We say that a valuation $\nu$ is \emph{finitary} if it maps every variable to a finite set of nodes.

	An \emph{\msou relabeling} is given by a tuple of \msou sentences $\Psi=(\psi_b)_{b\in\Bphabet}$ for a finite set $\Bphabet$.
	Suppose that we have a tree $T$ such that for every subtree $T\restr_\unode$ of $T$ there is exactly one $b\in\Bphabet$ for which $T\restr_\unode\models\psi_b$.
	The relabeling applied to such a tree produces a tree $\Psi(T)$ with the same domain as $T$, over the alphabet $\Bphabet$, 
	where every node $\unode\in\dom(T)$ gets labeled by that $b\in\Bphabet$ for which $T\restr_\unode\models\psi_b$
	(for trees $T$ not satisfying the above assumption, $\Psi(T)$ is undefined).
	
	\begin{theorem}\label{thm:main-2}
		Fix a finite alphabet $\Alphabet$ and a maximal arity $\rMax$.
		For every \msou formula $\phi$ there exists a formula $\phimso$ of \mso, and an \msou relabeling $\Psi$ 
		such that for every tree $T\in\Tt(\Alphabet,\rMax)$ and for every \emph{finitary} valuation $\nu$ in $T$,
		the tree $\Psi(T)$ is defined, and it is equivalent whether
		\begin{compactitem}
		\item	$T,\nu\models\phi$, and
		\item	$\Psi(T),\nu\models \phimso$.
		\end{compactitem}
		Moreover, $\FV(\phimso)\subseteq\FV(\phi)$.
	\end{theorem}
	
	\begin{remark}
		Both theorems above are constructive:
		knowing $\phi$, $\Alphabet$, and $\rMax$ one can compute either $\Omega$, or $\phimso$, respectively.
		The algorithm can be read out of our proof of existence, presented in the next sections.
	\end{remark}
	
	\begin{remark}
		Colcombet~\cite{forward-ramseyan} has shown that every formula $\phi$ of \mso can be rewritten into a formula $\phifo$ of first-order logic 
		referring to an \mso relabeling of a considered tree, analogously to our \cref{thm:main-2}
		(in his result, however, formulae used in the relabeling to relabel a node $\unode$ should be able to access the whole tree with the node $\unode$ marked, 
		instead of just the subtree rooted at $\unode$, as in our definition).
		Combining our \cref{thm:main-2} with this result,
		we can deduce that every \msou formula $\phi$ can be rewritten into a formula $\phifo$ of first-order logic 
                referring to an \msou relabeling of a considered tree (under the aforementioned extended definition of relabeling).
	\end{remark}

\section{Logical types and Theorem~\ref{thm:main-1}}

	In this section we prove our first result, \cref{thm:main-1}.
	To this end, we introduce logical types (aka.~phenotypes).
	These types contain more information than just the truth value of a formula.
	In consequence, types are compositional (as stated in \cref{lem:compositionality}),
	unlike truth values of formulae.

	In the sequel we assume that a finite alphabet $\Alphabet$ and a maximal arity $\rMax$ are fixed.
	Let $\phi$ be a formula of \msou, let $T$ be a tree, and let $\nu$ be a valuation.
	We define the \emph{$\phi$-type} of $T$ under valuation $\nu$, denoted $\pht{\phi}{T}{\nu}$, by induction on the size of $\phi$ as follows:
	\begin{compactitem}
	\item	if $\phi$ is of the form $a(\Xvar)$ (for some letter $a$) or $\Xvar\subseteq\Yvar$ then $\pht{\phi}{T}{\nu}$ is the logical value of $\phi$ in $T,\nu$, that is,
		$\true$ if $T,\nu\models\phi$ and $\false$ otherwise,
	\item	if $\phi$ is of the form $\Xvar\child_i \Yvar$, then $\pht{\phi}{T}{\nu}$ equals
		\begin{compactitem}
		\item	$\true$ if $T,\nu\models\phi$,
		\item	$\mathsf{empty}$ if $\nu(\Xvar)=\nu(\Yvar)=\emptyset$,
		\item	$\mathsf{root}$ if $\nu(\Xvar)=\emptyset$ and $\nu(\Yvar)=\{\epsilon\}$, and
		\item	$\false$ otherwise,
		\end{compactitem}
	\item	if $\phi\equiv(\psi_1\land\psi_2)$, then $\pht{\phi}{T}{\nu}=(\pht{\psi_1}{T}{\nu},\pht{\psi_2}{T}{\nu})$,
	\item	if $\phi\equiv(\neg\psi)$, then $\pht{\phi}{T}{\nu}=\pht{\psi}{T}{\nu}$,
	\item	if $\phi\equiv\exists\Xvar.\psi$ or $\phi\equiv\unbound \Xvar.\psi$, then
		\begin{align*}
			\pht{\phi}{T}{\nu}=\big(&\big\{\sigma\mid\exists X.\pht{\psi}{T}{\nu[\Xvar\mapsto X]}=\sigma\big\},\\
				&\big\{\sigma\mid\forall n\in\Nat.\exists X.\big(\pht{\psi}{T}{\nu[\Xvar\mapsto X]}=\sigma\land n\leq|X|<\infty\big)\big\}\big)\,,
		\end{align*}
		where $X$ ranges over sets of nodes of $T$.
	\end{compactitem}

	For each $\phi$, let $\Pht\phi$ denote the set of all potential $\phi$-types.
	Namely, $\Pht\phi=\{\true,\false\}$ in the first case,
	$\Pht\phi=\{\true, \mathsf{empty}, \mathsf{root}, \false\}$ in the second case,
	$\Pht\phi=\Pht{\psi_1}\times\Pht{\psi_2}$ in the third case,
	$\Pht\phi=\Pht{\psi}$ in the fourth case, and
	$\Pht\phi=(\Pp(\Pht\psi))^2$ in the fifth case.

	The following three propositions can be shown by a straightforward induction on the structure of a considered formula.
	
	\begin{proposition}\label{prop:pht-finite}
		For every \msou formula $\phi$ the set $\Pht\phi$ is finite.
	\qed\end{proposition}
	
	\begin{proposition}\label{prop:pht-2-form}
		For every \msou formula $\phi$ there is a function $\mathit{tv}_\phi\colon\Pht\phi\to\{\true,\false\}$ such that
		for every tree $T\in\Tt(\Alphabet,\rMax)$ and every valuation $\nu$ in $T$, it holds that $\mathit{tv}_\phi(\pht{\phi}{T}{\nu})=\true$ if and only if $T,\nu\models\phi$.
	\qed\end{proposition}
	
	In other words, the fact whether $\phi$ holds in $T,\nu$ is determined by $\pht{\phi}{T}{\nu}$.
	On the other hand, the $\phi$-type can be computed by an \msou formula:
	
	\begin{proposition}\label{prop:form-2-pht}
		For every \msou formula $\phi$ and every $\tau\in\Pht\phi$ there is an \msou formula $\psi_\tau$ with $\FV(\psi_\tau)\subseteq\FV(\phi)$ such that
		for every tree $T\in\Tt(\Alphabet,\rMax)$ and every valuation $\nu$ in $T$, it holds that $\pht{\phi}{T}{\nu}=\tau$ if and only if $T,\nu\models\psi_\tau$.
	\qed\end{proposition}

	Next, we observe that types behave in a compositional way, as formalized below.

	\begin{lemma}\label{lem:compositionality}
		For every letter $a$, every $r\in\Nat$, and every \msou formula $\phi$,
		one can compute a function $\Comp_{a,r,\phi}\colon\Pp(\FV(\phi))\times(\Pht\phi)^r\to\Pht\phi$ such that
		for every tree $T$ whose root has label $a$ and $r$ children, and for every valuation $\nu$,
		\begin{align*}
			\pht{\phi}{T}{\nu}=\Comp_{a,r,\phi}(\{\Xvar\in\FV(\phi)\mid\epsilon\in\nu(\Xvar)\},\pht{\phi}{T\restr_1}{\nu\restr_1},\dots,\pht{\phi}{T\restr_r}{\nu\restr_r})\,.
		\end{align*}
	\end{lemma}

	\begin{proof}
		We proceed by induction on the size of $\phi$.

		When $\phi$ is of the form $b(\Xvar)$ or $\Xvar\subseteq \Yvar$, then we see that $\phi$ holds in $T,\nu$ if and only if it holds in every subtree $T\restr_i,\nu\restr_i$ and in the root of $T$.
		Thus, for $\phi\equiv b(\Xvar)$ as $\Comp_{a,r,\phi}(R,\tau_1,\dots,\tau_r)$ we take $\true$ when $\tau_i=\true$ for all $i\in\{1,\dots,r\}$ and either $a=b$ or $\Xvar\not\in R$.
		For $\phi\equiv(\Xvar\subseteq \Yvar)$ the last part of the condition is replaced by ``if $\Xvar\in R$ then $\Yvar\in R$''.

		Next, suppose that $\phi\equiv(\Xvar\child_k \Yvar)$.
		Then as $\Comp_{a,r,\phi}(R,\tau_1,\dots,\tau_r)$ we take
		\begin{compactitem}
		\item	$\true$ if $\tau_j=\true$ for some $j\in\{1,\dots,r\}$, and $\tau_i=\mathsf{empty}$ for all $i\in\{1,\dots,r\}\setminus\{j\}$, and $\Xvar\not\in R$, and $\Yvar\not\in R$,
		\item	$\true$ also if $\tau_k=\mathsf{root}$, and $\tau_i=\mathsf{empty}$ for all $i\in\{1,\dots,r\}\setminus\{k\}$, and $\Xvar\in R$, and $\Yvar\not\in R$,
		\item	$\mathsf{empty}$ if $\tau_i=\mathsf{empty}$ for all $i\in\{1,\dots,r\}$, and $\Xvar\not\in R$, and $\Yvar\not\in R$,
		\item	$\mathsf{root}$ if $\tau_i=\mathsf{empty}$ for all $i\in\{1,\dots,r\}$, and $\Xvar\not\in R$, and $\Yvar\in R$, and
		\item	$\false$ otherwise.
		\end{compactitem}
		By comparing this definition with the definition of the type we immediately see that the thesis is satisfied.

		When $\phi\equiv(\neg\psi)$, we simply take $\Comp_{a,r,\phi}=\Comp_{a,r,\psi}$, 
		and when $\phi\equiv(\psi_1\land\psi_2)$, as $\Comp_{a,r,\phi}(R,(\tau_1^1,\tau_1^2),\dots,(\tau^1_r,\tau^2_r))$ 
		we take the pair of $\Comp_{a,r,\psi_i}(R\cap\FV(\phi_i),\tau_1^i,\dots,\tau^i_r)$ for $i\in\{1,2\}$.

		Finally, suppose that $\phi\equiv\exists\Xvar.\psi$ or $\phi\equiv\unbound\Xvar.\psi$.
		The arguments of $\Comp_{a,r,\phi}$ are pairs $(\tau_1,\rho_1),\dots,\allowbreak(\tau_r,\rho_r)$.
		Let $A$ be the set of tuples $(\sigma_1,\dots,\sigma_r)\in\tau_1\times\dots\times\tau_r$,
		and let $B$ be the set of tuples $(\sigma_1,\dots,\sigma_r)$ such that $\sigma_j\in\rho_j$ for some $j\in\{1,\dots,r\}$ and $\sigma_i\in\tau_i$ for all $i\in\{1,\dots,r\}\setminus\{j\}$.
		As $\Comp_{a,r,\phi}(R,(\tau_1,\rho_1),\dots,(\tau_r,\rho_r))$ we take
		\begin{align*}
			(&\{\Comp_{a,r,\psi}(R\cup\{\Xvar\},\sigma_1,\dots,\sigma_r),\Comp_{a,r,\psi}(R\setminus\{\Xvar\},\sigma_1,\dots,\sigma_r)\mid(\sigma_1,\dots,\sigma_r)\in A\},\\
			 &\{\Comp_{a,r,\psi}(R\cup\{\Xvar\},\sigma_1,\dots,\sigma_r),\Comp_{a,r,\psi}(R\setminus\{\Xvar\},\sigma_1,\dots,\sigma_r)\mid(\sigma_1,\dots,\sigma_r)\in B\})\,.
		\end{align*}
		The two possibilities, $R\cup\{\Xvar\}$ and $R\setminus\{\Xvar\}$, correspond to the fact that when quantifying over $\Xvar$,
		the root of $T$ may be either taken to the set represented by $\Xvar$ or not.
		The second coordinate is computed correctly due to the pigeonhole principle: if for every $n$ we have a set $X_n$ of cardinality at least $n$ (satisfying some property),
		then we can choose an infinite subsequence of these sets such that either the root belongs to all of them or to none of them,
		and one can choose some $j\in\{1,\dots,r\}$ such that the sets contain unboundedly many descendants of $j$.
	\end{proof}
	
	Now \cref{thm:main-1} follows easily:
	
	\begin{proof}[Proof of \cref{thm:main-1}]
		Let $\phi$ be the \msou formula under consideration.
		We should define a set $\Omega$ of tuples of \msou formulae.
		To this end, consider the set $\Phi\subseteq\Pht\phi$ containing those $\phi$-types for which $\phi$ is true, 
		that is, $\phi$-types $\tau$ such that $\mathit{tv}_\phi(\tau)=\true$, where $\mathit{tv}_\phi$ is the function defined in \cref{prop:pht-2-form}.
		Next, for every $\phi$-type $\tau\in\Phi$, for every letter $a\in\Alphabet$ (root's label), and for every $r\in\set{0,\dots,\rMax}$ (number of root's children)
		consider all tuples $(R,\tau_1,\dots,\tau_r)\in\Pp(\FV(\phi))\times(\Pht\phi)^r$ 
		such that $\Comp_{a,r,\phi}(R,\tau_1,\dots,\tau_r)=\tau$, where $\Comp_{a,r,\phi}$ is the function defined in \cref{lem:compositionality}.
		For every such a tuple, we add to $\Omega$ a tuple $(\eta_{a,R},\psi_{\tau_1},\dots,\psi_{\tau_r})$, where
		\begin{align*}
			\eta_{a,R}\equiv
				\forall\Yvar.\,a(\Yvar)\land
				\bigwedge\nolimits_{\Xvar\in R}(\forall\Yvar.\;\Yvar\subseteq\Xvar)\land
				\bigwedge\nolimits_{\Xvar\in\FV(\phi)\setminus R}\neg (\forall\Yvar.\;\Yvar\subseteq\Xvar)
		\end{align*}
		and where $\psi_{\tau_i}$ are the formulae corresponding to types $\tau_i$, as defined in \cref{prop:form-2-pht}.
		Because there are finitely many possibilities for $a$, $r$, and $R$, and finitely many $\phi$-types (cf.~\cref{prop:pht-finite}), the set $\Omega$ is finite.
		
		Consider now a particular tree $T\in\Tt(\Alphabet,\rMax)$, and a valuation $\nu$ in $T$.
		Let $a$ be the label of the root of $T$, let $r$ be the number of root's children, and let $R=\set{\Xvar\in\FV(\phi)\mid\epsilon\in\nu(\Xvar)}$.
		Moreover, let $\tau=\pht{\phi}{T}{\nu}$, and for $i\in\set{1,\dots,r}$ let $\tau_i=\pht{\phi}{T\restr_i}{\nu\restr_i}$.
		By \cref{lem:compositionality} we have that $\tau=\Comp_{a,r,\phi}(T,\tau_1,\dots,\tau_r)$.
		
		Suppose first that $T,\nu\models\phi$.
		By \cref{prop:pht-2-form} this implies that $\tau\in\Phi$, and hence $(\eta_{a,R},\psi_{\tau_1},\dots,\psi_{\tau_r})\in\Omega$.
		We see that $\roottree(T),\roottree(\nu)\models\eta_{a,R}$, and that $T\restr_i,\nu\restr_i\models\psi_{\tau_1}$ for all $i\in\set{1,\dots,r}$ (by \cref{prop:form-2-pht}),
		which gives the thesis.
		
		Conversely, suppose that for some tuple $(\eta_{a',R'},\psi_{\tau_1'},\dots,\psi_{\tau_r'})\in\Omega$ it is the case that
		$\roottree(T),\roottree(\nu)\models\eta_{a',R'}$ and $T\restr_i,\nu\restr_i\models\psi_{\tau_1'}$ for all $i\in\set{1,\dots,r}$.
		We see that necessarily $a'=a$, $R'=R$, and $\tau_i'=\tau_i$ for all $i\in\set{1,\dots,r}$ (by \cref{prop:form-2-pht}).
		Thus, actually $(\eta_{a,R},\psi_{\tau_1},\dots,\psi_{\tau_r})\in\Omega$, which implies that $\tau\in\Phi$, and in consequence $T,\nu\models\phi$, by \cref{prop:pht-2-form}.
	\end{proof}

\section{Proof of Theorem~\ref{thm:main-2}}

	In this section we prove our second result, \cref{thm:main-2}.
	Recall that our goal is to decompose a formula $\phi$ into an \msou relabeling $\Psi$ and an \mso formula $\phimso$,
	assuming that free variables of $\phi$ are valuated to finite sets.
	
	The idea here is that the finite top part of a tree, where all the set variables are valuated, can be handled by \mso 
	(intuitively: in a finite part nothing can be unbounded, so the \unbound quantifier is void here, and hence it can be eliminated).
	The remaining part of the tree consists of subtrees in which all variables are valuated to empty sets;
	the $\phi$-type of every such a subtree is fixed, so it can be precomputed and written in the label of the root of that subtree.
	
	Again, in this section we assume that $\Alphabet$ and $\rMax$ are fixed.
	Let $\VALempty$ be the \emph{empty valuation}, mapping every variable to the empty set.
	We need formulae computing $\phi$-types under the assumption that the valuation is empty.

	\begin{proposition}\label{prop:form-2-pht-empty}
		For every \msou formula $\phi$ and every $\tau\in\Pht\phi$ there is an \msou sentence $\psi_\tau^\emptyset$ such that
		for every tree $T\in\Tt(\Alphabet,\rMax)$ it holds that $\pht{\phi}{T}{\VALempty}=\tau$ if and only if $T\models\psi_\tau^\emptyset$.
	\end{proposition}
	
	\begin{proof}
		By \cref{prop:form-2-pht} we have a formula $\psi_\tau$ checking that the $\phi$-type is $\tau$ under a given valuation.
		We obtain $\psi_\tau^\emptyset$ from $\psi_\tau$ by making a conjunction with statements of the form $\mathit{empty}(\Xvar)$ for all free variables $\Xvar$, 
		and then surrounding the formula by quantifiers $\exists\Xvar$ for all free variables $\Xvar$.
	\end{proof}
	
	We now prove \cref{thm:main-2}.
	
	\begin{proof}[Proof of \cref{thm:main-2}]
		Beside of the sentences $\psi_\tau^\emptyset$ from \cref{prop:form-2-pht-empty},
		for every $a\in\Alphabet$ we consider a sentence $\eta_a$ saying that the root of a tree is labeled by $a$.
		As the relabeling we take $\Psi=(\eta_a\cap\psi_\tau^\emptyset)_{(a,\tau)\in\Alphabet\times\Pht\phi}$.
		
		The formula $\phimso$ starts with a sequence of $|\Pht{\phi}|$ existential quantifiers,
		quantifying over variables $\Xvar_\tau$ for all $\tau\in\Pht{\phi}$.
		The intention is that, in a tree $T$, every $\Xvar_\tau$ represents the set of nodes $\unode$ such that $\pht{\phi}{T\restr_\unode}{\nu\restr_\unode}=\rho$.
		Inside the quantification we say that
		\begin{compactitem}
		\item	the sets represented by these variables are disjoint, and every node belongs to some of them,
		\item	the root belongs to $\Xvar_\tau$ for some $\tau$ such that $\mathit{tv}_\phi(\tau)=\true$, where $\mathit{tv}_\phi$ is the function defined in \cref{prop:pht-2-form},
		\item	if a node with label $(a,\tau_\emptyset)$ belongs to $\Xvar_\tau$, and its children belong to $\Xvar_{\tau_1},\dots,\Xvar_{\tau_r}$, respectively
			(where $r\leq\rMax$), and $R$ is the set of free variables $\Yvar$ of $\phi$ for which the node belongs to $\nu(\Yvar)$, then
			$\tau=\Comp_{a,r,\phi}(R,\tau_1,\dots,\tau_r)$
			(there are only finitely many possibilities for $\tau,\tau_\emptyset,\tau_1,\dots,\tau_r\in\Pht\phi$, for $r\in\{0,\dots,\rMax\}$, for $a\in\Alphabet$,
			and finitely many free variables of $\phi$, thus the constructed formula can be just a big alternative listing all possible cases), and
		\item	if a node with label $(a,\tau_\emptyset)$ belongs to $\Xvar_\tau$ and none of $\nu(\Yvar)$ for $\Yvar$ free in $\phi$ contains this node or some its descendant,
			then $\tau=\tau_\emptyset$.
		\end{compactitem}
		Consider now a tree $T\in\Tt(\Alphabet,\rMax)$ and a valuation $\nu$ in this tree.
		If $\pht{\phi}{T}{\nu}=\tau$, then we can show that $\phimso$ is true
		by taking for $\Xvar_\tau$ the set of nodes $\unode$ for which $\pht{\phi}{T\restr_\unode}{\nu\restr_\unode}=\tau$ (for every $\tau\in\Pht{\phi}$).
		Conversely, suppose that $\phimso$ is true.
		Then we can prove that a node $\unode$ can belong to the set represented by $\Xvar_\tau$ (for $\tau\in\Pht{\phi}$) only when $\pht{\phi}{T\restr_\unode}{\nu\restr_\unode}=\tau$.
		The proof is by a straightforward induction on the number of descendants of $\unode$ that belong to $\nu(\Yvar)$ for some $\Yvar$ free in $\phi$;
		we use \cref{lem:compositionality} for the induction step.
	\end{proof}

\bibliography{bib}

\end{document}